\documentclass[12pt]{article}
\usepackage{amsfonts}
\usepackage{amssymb}
\usepackage{color}
\usepackage{graphicx}

\newcommand{\F}{{\cal F}}
\newcommand{\Sc}{{\cal S}}

\newcommand{\be}{\begin{equation}}
\newcommand{\en}{\end{equation}}
\newcommand{\bea}{\begin{eqnarray}}
\newcommand{\ena}{\end{eqnarray}}
\newcommand{\beano}{\begin{eqnarray*}}
\newcommand{\enano}{\end{eqnarray*}}

\newcommand{\M}{{\cal M}}

\newcommand{\V}{{\cal V}}
\newcommand{\ltwo}{{\Lc^2(\mathbb{R})}}

\renewcommand{\l}{\langle}
\renewcommand{\r}{\rangle}
\newcommand{\pin}[2]{\l#1 , #2\r}

\newcommand{\Lc}{{\cal L}}

\newcommand{\un}{{\underline n}}

\newtheorem{thm}{Theorem}
\newtheorem{cor}[thm]{Corollary}

\newtheorem{prop}[thm]{Proposition}
\newtheorem{defn}[thm]{Definition}

\newenvironment{proof}{\noindent {\bf Proof:}}{\hfill$\Box$}

\catcode `\@=11 \@addtoreset{equation}{section}

\catcode `\@=12
\begin{document}

\begin{center}
{\Large \textbf{$e$-product of distributions, with applications}} \vspace{2cm%
}\\[0pt]

{\large F. Bagarello}
%\footnote[1]{ Dipartimento di Matematica ed Applicazioni,
%Fac.Ingegneria, Universit\`a di Palermo, I - 90128  Palermo, Italy}
\vspace{3mm}\\[0pt]
Dipartimento di Ingegneria,\\[0pt]
Universit\`{a} di Palermo, I - 90128 Palermo,\\
and I.N.F.N., Sezione di Catania\\
E-mail: fabio.bagarello@unipa.it\\

\vspace{7mm}

\end{center}

\vspace*{2cm}

\begin{abstract}
\noindent We consider and reformulate a recent definition of multiplication between distributions. We show that this definition can be adopted, in particular, to prove biorthonormality of some distributions arising when looking to the (generalized) eigenvalues of a specific non self-adjoint number-like operator, considered in connection with the recently introduced {\em weak pseudo-bosons}. Several examples are discussed in details.
\end{abstract}

\vspace*{1cm}


\vfill

\newpage

% Section 1

\section{Introduction}

The role of Hilbert spaces in quantum mechanics is, of course, quite relevant, mainly in view of the probabilistic interpretation of the wave function $\Psi(x,t)$ of the given physical system $\Sc$. However, already several years ago, people started to be interested in describing situations in which Hilbert spaces, and $\ltwo$  in particular, are not enough. In these cases, it became useful to consider the so-called {\em rigged Hilbert spaces}, \cite{rhs}, or to work with {\em partial-inner-product} (PIP) spaces, \cite{pip}. In particular, the main idea for this latter extension is that the scalar product on, say, $\ltwo$, can be extended to other pairs of functions, $f(x)$ and $g(x)$, whenever $\overline{f(x)}\,g(x)\in\Lc^1(\mathbb{R})$. In this way we can still consider $\int_{\mathbb{R}}\overline{f(x)}\,g(x)\,dx$, and call it {\em an extended scalar product} between the two functions. This is the case, for instance, when $f(x)\in\Lc^p(\mathbb{R})$ and $g(x)\in\Lc^q(\mathbb{R})$, with $p^{-1}+q^{-1}=1$.

It is also very well known, since a long time, that distributions play a relevant role in some concrete physical situations. This was discussed, for instance, in \cite{wig} in the context of quantum field theory: quantum fields are operator-valued distributions, and this is strongly related to the divergences arising in concrete computations in quantum field theory, and to the need for a renormalization of the theory, see \cite{itz,bd,muta}, for instance.

Distributions do not appear only in quantum field theory. Indeed, they also play a role in quantum mechanics. For instance, they are used to model {\em point interactions}, see \cite{teta} and references therein to get the feeling of the relevance of this hot topics, mainly in mathematical physics (but not only).

In this paper we are motivated by some recent results connected to {\em weak pseudo-bosons} and, more generally, to quantum mechanics for non self-adjoint Hamiltonians, \cite{bagspringer}. In particular, already in 2020 we started to be interested to the role of distributions in concrete phyical systems, and in particular to a specific ladder structure (i.e., to some annihilation and creation operators) in the analysis of the position and the momentum operators $\hat x$ and $\hat p=-i\frac{d}{dx}$, \cite{bag2020JPA}. What we proved is that it is possible to use these operators to construct a family of eigenstates of a certain number-like operator, $\hat N=i\hat x\hat p$, but these (generalized) eigenstates are indeed distributions. This was also linked to what discussed in \cite{bagrocgar2019}, where we proved that the vacuum of the quantized damped harmonic oscillator is, indeed, a Dirac delta function.  The role of distributions in quantum mechanics for weak pseudo-bosons was further considered in \cite{bag2022,bag2023}. 
In all these contributions we had to consider the problem of introducing a sort of scalar product for certain  generalized eigenstates of some {\em number-like operators}, relevant for our systems. This was useful to extend, to a distributional settings, the notion of biorthogonality of vectors. More specifically, we first used a definition of this extended scalar product in terms of convolutions, \cite{bag2020JPA}. Later, and in particular in \cite{bag2023}, we proposed a different definition which we believe is more naturally connected to our quantum mechanical background, and which works well in some situations which are interesting for us. This last approach is the main interest of our analysis here.

The paper is organized as follows: in Section \ref{sect2} we will briefly review the two definitions of multiplication of distributions mentioned above, and already used in previous papers. In Section \ref{sect3} we focus on what we call {\em the $e$-product}, in a slightly new version, and we discuss some general results for it, including some preliminary facts on a adjoint map connected to the $e$-product. Section \ref{sect4} contains several examples, some of them positive (i.e., the $e$-product exists), and some negative (i.e., the $e$-product does not exist), while our conclusions are given in Section \ref{sect5}.

\section{Preliminaries}\label{sect2}

In \cite{vlad} a possible way to introduce a multiplication between distributions was discussed. This method was used in some recent applications, as, for instance, in \cite{bagspringer,bag2020JPA}, and it is based on the fact that the scalar product between two {\em good } functions $f(x)$ and $g(x)$, for instance $f(x),g(x)\in\Sc(\mathbb{R})$, the Schwartz space\footnote{For reader's convenience we recall that $\Sc(\mathbb{R})=\{f(x)\in C^\infty(\mathbb{R}):\, \lim_{|x|,\infty}|x|^kf^{(l)}(x)=0, \, \forall k,l=0,1,2,\ldots\}$.}, can be written in terms of a convolution between $\overline{f(x)}$ and $\tilde{g}(x)=g(-x)$: $\left<f,g\right>=(\overline{f}* \tilde{g})(0)$. Hence it is natural to define the scalar product between two tempered distributions, $F(x), G(x)\in\Sc'(\mathbb{R})$\footnote{ We recall that $\Sc'(\mathbb{R})$ is the set of  all continuous functionals on $\Sc(\mathbb{R})$, \cite{vlad}.}, as the following convolution:
\be
\left<F,G\right>=(\overline{F}* \tilde{G})(0),
\label{21}\en
whenever this convolution exists, which is not always true. Notice that, in order to compute $\left<F,G\right>$, it is first necessary to compute $(\overline{F}* \tilde{G})[f]$, $f(x)\in\Sc(\mathbb{R})$, and this can be done by using the equality $(\overline{F}* \tilde{G})[f]=\left<F,G*f\right>$, i.e. the action of $(\overline{F}* \tilde{G})(x)$ on the function $f(x)$, which, again, is not always well defined.

This approach has been used in some concrete situations in recent years, mainly to check if the generalized eigenstates of some non self-adjoint operator $\hat H$ are biorthonormal (with respect to this generalized product) to those of $\hat H^\dagger$. However, in \cite{bag2023}, we showed that this approach is not really the most convenient when dealing with a specific gain-loss system relevant for the analysis of the damped harmonic oscillator. For this reason, we proposed a different definition of {\em extended} scalar product, which refers to a fixed orthonormal basis (ONB) $\F_e$ in the Hilbert space where the physical system is defined.

Let us summarize some of the results discussed in \cite{bag2023}, using as our Hilbert space the space $\ltwo$.

 First of all we introduce an ONB in $\ltwo$,
$\F_e=\{e_{n}(x)\in \Sc(\mathbb{R}), \, n\geq0\}$. For instance,  the set $\F_e$ could be the set of the eigenstates of the quantum harmonic oscillator, see (\ref{41}) below. Using the Parceval identity we have that, for all $f(x), g(x)\in \ltwo$, 
$$
\pin{f}{g}=\sum_{n}\pin{f}{e_n}\pin{e_n}{g}=\sum_{\un} \overline{f}[e_n]\,g[\overline{e_n}]=\sum_{n} \overline{f}[e_n]\,g[{e_n}],
$$
assuming also that each $e_n(x)$ is real, to simplify the notation. Here we use
\be
\overline{h}[c]=\int_{\mathbb{R}} \overline{h(x)}\,c(x)\, dx=\pin{h}{c}.
\label{22}\en
 In the Parceval identity above the particular choice of $\F_e$ is not relevant, as far as $f(x), g(x)\in \ltwo$: in other words, if $\F_c=\{c_{n}(x)\in \ltwo, \, n\geq0\}$ is a different ONB of $\ltwo$, with elements not necessarily in $\Sc(\mathbb{R})$, we have
 $$
 \sum_{n}\pin{f}{e_n}\pin{e_n}{g}=\sum_{n}\pin{f}{c_n}\pin{c_n}{g},
 $$
 for all $f(x),g(x)\in\ltwo$. 
 
 Now, since $e_n(x)\in\Sc(\mathbb{R})$, the following quantity is well defined:
\be
\overline{K}[e_n]=\pin{K}{e_n},
\label{23}\en
 for all $K(x)\in\Sc'(\mathbb{R})$.
What might exist, or not, is the following sum
\be
\pin{F}{G}_e=\sum_{n}\overline{F}[e_n]G[e_n],
\label{24}\en
where $F(x), G(x)\in\Sc'(\mathbb{R})$ are two generic tempered distributions. For this reason, in \cite{bag2023} we have proposed the following:

\begin{defn}\label{def1}
	Two tempered distributions $F(x), G(x)\in\Sc'(\mathbb{R})$ are $\F_e$-multiplicable if the series in (\ref{24}) converges. In this case, we call $\pin{F}{G}_e$ the $e$-product of $F(x)$ and $G(x)$.
\end{defn}

It is obvious that for all the square integrable functions the $e$-product exists, and that the result is independent of the specific choice of $\F_e$. Then Definition \ref{def1} makes sense on a large set of tempered distributions, all those defined by ordinary square-integrable functions. Moreover, at least for these functions, (\ref{24}) and (\ref{21}) coincide, since they both coincide with their ordinary scalar product in $\ltwo$. We have already shown in \cite{bag2023} that $\pin{F}{G}_e$ is also well defined in other cases. We will discuss many other (positive and negative) examples in Section \ref{sect4}.

The following results, proved in \cite{bag2023}, are natural extensions of the properties of any { ordinary} scalar product to   $\pin{.}{.}_e$.

\vspace{2mm}

{\bf Result $\sharp 1$: } If $F(x), G(x)\in\Sc'(\mathbb{R})$ are such that $\pin{F}{G}_e$ exists, then also $\pin{G}{F}_e$ exists and
\be
\pin{F}{G}_e=\overline{\pin{G}{F}_e}.
\label{25}\en

\vspace{2mm}

{\bf Result $\sharp 2$: }  
If $F(x), G(x), L(x)\in\Sc'(\mathbb{R})$ are such that $\pin{F}{G}_e$ and $\pin{F}{L}_e$ exist, then also $\pin{F}{\alpha G+\beta L}_e$ exists, for all $\alpha,\beta\in\mathbb{C}$, and
\be
\pin{F}{\alpha G+\beta L}_e=\alpha\,\pin{F}{ G}_e+\beta\pin{F}{L}_e.
\label{26}\en
Then the $e$-product is linear in the second variable. Of course, in view of (\ref{25}), the $e$-product is anti-linear in the first variable.

\vspace{2mm}

{\bf Result $\sharp 3$: } 
If $F(x)\in\Sc'(\mathbb{R})$ is such that $\pin{F}{F}_e$ exists, then $\pin{F}{F}_e\geq0$. In particular, if $\pin{F}{F}_e=0$, then $F[f]=0$ for all $f(x)\in\Lc_e$, the linear span of the $e_n(x)$'s.

\vspace{2mm}

In what follows we will slightly reconsider Definition \ref{def1} and extend these results, in a way which appears more useful for future applications and for a more detailed analysis.

\section{The $e$-product}\label{sect3}

We start considering a set $\V$, dense in $\ltwo$, endowed with a topology $\tau_\V$ which makes of $\V$ a complete set of functions. Let $\V'$ be its dual set, i.e. the set of the linear continuous functionals on $\V$. Hence we have $\V\subseteq\ltwo\subseteq\V'$, where we are also admitting the possibility that these sets coincide.

 Of course, $\Sc(\mathbb{R})$ is a possible example of the set $\V$, with its own topology $\tau_\Sc$, but it is not the only one. For instance, in \cite{bag2022} we have introduced another such example, in connection with the inverted quantum harmonic oscillator. We will now summarize some of the results which are relevant for us. All the details and the proofs of the claims that follow can be found in \cite{bag2022}. 

\vspace{2mm}

{\bf Example} (An interesting settings):--
Let $\rho(x)$ be a strictly positive, L-measurable, function, and let us introduce the set
\be
\V_\rho=\left\{f(x)\in\ltwo:\,\rho(x)f(x)\in\ltwo\right\}.
\label{31}\en
If $\rho(x)\in\Lc^\infty(\mathbb{R})$, then $\V_\rho=\ltwo$.
Also, if $\rho(x)$ is continuous, not necessarily bounded, $\V_\rho$ is dense in $\ltwo$.

We see that $\V_\rho$ is a significantly large set, in many interesting situations. We endow now $\V_\rho$ with a topology, $\tau_\rho$, defined as follows:

\begin{defn} Given a sequence $\{f_n(x)\in\V_\rho\}$, we say that this is $\tau_\rho$-convergent if (i) $\{f_n(x)\in\V_\rho\}$ is $\|.\|$-Cauchy and if (ii) $\{\rho(x)f_n(x)\in\V_\rho\}$ is $\|.\|$-Cauchy.
\end{defn}

Hence, see \cite{bag2022}, 
	if $\rho^{-1}(x)\in\Lc^\infty(\mathbb{R})$, then $\V_\rho$ is closed with respect to $\tau_\rho$.

Restricting to real $\rho(x)$ satisfying the above assumption we further introduce another set of functions, $\Theta_\rho$, as follows:
\be
\Theta_\rho=\{\Phi(x), \mbox{L-measurable: } \Phi(x)\rho^{-1}(x)\in\ltwo\}.
\label{32}\en
It is possible to show that every element of $\Theta_\rho$ defines a continuous linear functional on $\V_\rho$, that is, an  element of $\V_\rho'$. In other words, $\Theta_\rho\subseteq\V_\rho'$. In this way we have constructed a second example of the functional settings which is useful for us. Notice that, despite of what happens for $\Sc(\mathbb{R})$ and $\Sc'(\mathbb{R})$, we are here considering only functions. 

\vspace{2mm}

Going back to the general situation, suppose now that $\F_e=\{e_n(x)\in\V,\,n\geq0\}$ is an ONB for $\ltwo$. As in (\ref{23}) we put
\be
\overline{K}[e_n]=\pin{K}{e_n},
\label{33}\en
for all $K(x)\in\V'$, which is well defined for all $n$, of course. As in Definition \ref{def1} we say that, taken $F(x), G(x)\in\V'$, their $e$-product exists if the following series converge:
\be
\pin{F}{G}_e=\sum_{n}\overline{F}[e_n]G[e_n]=\sum_n\pin{F}{e_n}\,\pin{e_n}{G},
\label{34}\en
where we are assuming again that $e_n(x)$ is a real function. We will say that $(F,G)\in\M_e$. In other words, $\M_e$ is the subset of $(\V',\V')$ consisting of all the pairs of elements of $\V'$ for which the $e$-product exists. As already observed, $(f(x),g(x))\in\M_e$, independently of the choice of $\F_e$ and $\forall f(x),g(x)\in\ltwo$. Moreover, similarly to what observed before, 
\begin{equation}
\pin{f}{g}_e=\pin{f}{g}_c,
\label{35}\end{equation}
for all $f(x),g(x)\in\ltwo$ and for all ONB $\F_e$ and $\F_c$.

Even if in Section \ref{sect4} we will show in several explicit examples that the $e$-product indeed exists outside $\ltwo$, it may be interesting to show a somehow general result concerning the existence of this product. For that we first recall what a Bessel sequence is. Let $\F_\varphi=\{\varphi_n(x)\in\ltwo\}$ be a set of square-integrable functions. We say that $\F_\varphi$ is a Bessel function if a $B>0$ exists, $B<\infty$, such that
\be
\sum_n|\pin{\varphi_n}{f}|^2\leq B\|f\|^2,
\label{36}
\en
$\forall f(x)\in\ltwo$, \cite{heil}. Suppose now that $\varphi_n(x)=\rho(x)e_n(x)$, where $\rho(x)$ was introduced in the example above for $\V_\rho$. We assume that $e_n(x)\in\V_\rho$ for all $n$. First of all, it is possible to check that $\sum_n\pin{f}{\varphi_n}\pin{\varphi_n}{g}$ converges for all $f(x), g(x)\in\ltwo$. Indeed we have
$$
\left|\sum_n\pin{f}{\varphi_n}\pin{\varphi_n}{g}\right|\leq \sum_n|\pin{f}{\varphi_n}||\pin{\varphi_n}{g}|\leq
\sqrt{\sum_n|\pin{f}{\varphi_n}|^2}\sqrt{\sum_n|\pin{\varphi_n}{g}|^2}\leq
 B\|f\|\|g\|,
$$
because of (\ref{36}) and of the Schwarz inequality. Now, taking $\Phi_1(x), \Phi_2(x)\in\Theta_\rho$, it follows that $\Phi_1(x)\rho^{-1}(x), \Phi_2(x)\rho^{-1}(x)\in\ltwo$. Then we have that $\sum_n\pin{\Phi_1\rho^{-1}}{\varphi_n}\pin{\varphi_n}{\Phi_1\rho^{-1}}$ converges. But this implies that $\sum_n\pin{\Phi_1}{e_n}\pin{e_n}{\Phi_1}$ converges as well. Hence we can define
\be
\pin{\Phi_1}{\Phi_2}_e=\sum_n\pin{\Phi_1}{e_n}\pin{e_n}{\Phi_1}=\sum_n\pin{\Phi_1\rho^{-1}}{\varphi_n}\pin{\varphi_n}{\Phi_1\rho^{-1}},
\en
which means that, in this case, $\Phi_1$ and $\Phi_2$ admit $e$-product.

Results $\sharp 1$, $\sharp 2$ and $\sharp 3$ of Section \ref{sect2} can be restated in the present settings, where $\Sc(\mathbb{R}), \Sc'(\mathbb{R})$ are replaced by $\V,\V'$. In particular, if $F(x)\in\V'$ is such that $\pin{F}{F}_e=0$, then, as in Section \ref{sect2}, $F[f]=0$ for all $f(x)\in\Lc_e$, the linear span of the $e_n(x)$'s. But still we cannot conclude that $F=0$, the reason being the lack of continuity (in the right topology) of $F$. For this reason it may be convenient to introduce here the following definition.

\begin{defn}
	(i) A set $\V_0\subseteq\V$ is called {\em $e$-separating} if, taken $\F\in\V'$ such that $\pin{F}{\varphi}_e=0$, $\forall\varphi\in\V_0$, then $F=0$.
	
	(ii) The ONB $\F_e$ is called {\em $\F_e$-separating} if, taken $F\in\V'$ such that $\pin{F}{e_n}=0$, $\forall n$, then $F=0$.
	
\end{defn}

Notice that, while in $(i)$ we are using $\pin{.}{.}_e$, in $(ii)$ the original scalar product $\pin{.}{.}$ is involved. Notice also that, if $F\in\V'$ is such that $\pin{F}{e_n}=0$, $\forall n$, then $\pin{F}{F}_e$ obviously exists, and its zero, and vice-versa.

It is now possible, and not surprising, to check that these two notions of separability are indeed connected. In fact, calling $\V_0$ as the set of (finite) linear combinations of the $e_n(x)$'s, $\V_0=l.s.\{e_n\}$, we have the following:

\begin{prop}
	If $\F_e$ is {\em $\F_e$-separating}, then $\V$ is $e$-separating.	
	Viceversa,  if $\V_0$ is $e$-separating, then $\F_e$ is {\em $\F_e$-separating}.
\end{prop}

\begin{proof}
	Suppose first that $\F_e$ is {\em $\F_e$-separating}. We want to check that, if $F\in\V'$ is such that $\pin{F}{\varphi}_e=0$, $\forall\varphi\in\V_0$, then $F=0$. But this is clear, since our assumption means, in particular, that $\pin{F}{e_n}_e=0$, $\forall n$. It is now trivial to check that, $\forall G\in\V'$, the following equality holds:
	\be 
	\pin{G}{e_n}_e=\pin{G}{e_n},
	\label{37}
	\en 
	for all $n$. Hence we deduce that $\pin{F}{e_n}=0$, $\forall n$, so that $F=0$, using the fact that $\F_e$ is {\em $\F_e$-separating}.
	
	Let us now assume that $\V_0$ is $e$-separating, and let $F\in\V'$ be such that $\pin{F}{e_n}=0$, $\forall n$. We now want to prove that $F=0$. Again, due to (\ref{37}), we also have that $\pin{F}{e_n}_e=0$, $\forall n$, and therefore $\pin{F}{\varphi}_e=0$ for all $\varphi\in\V_0$. But $\V_0$ is $e$-separating. Hence $F=0$.

\end{proof}

As a corollary the following implications follow:

\begin{cor}
	The set $\V_0$ is $e$-separating if and only if $\F_e$ is {\em $\F_e$-separating}.
\end{cor}

It is now clear that the following is true:

if $\V_0$ is $e$-separating, and if $\pin{F}{\varphi}_e=\pin{G}{\varphi}_e$ for all $\varphi\in\V_0$, then $F=G$;

and, if $\F_e$ is $\F_e$-separating, and if $\pin{F}{e_n}_e=\pin{G}{e_n}_e$  $\forall n$, then $F=G$.

\subsection{Something about the adjoint}

The last results we deduced above are useful in connection with the possibility to define an adjoint map connected with the $e$-product. What we are going to discuss here are just few preliminary results, which however already show that something interesting can be deduced.

Given a linear operator $X$, which we assume to be defined on all of $\V'$, we would like to give a meaning to a new operator, $X^\ddagger$, defined by the following formula:
\be
	\pin{X^\ddagger \Phi}{\varphi}_e=\pin{\Phi}{X\varphi}_e,
\label{38}\en
where $\varphi\in\V_0$ (or $\varphi\in\V$) and $\Phi\in\V'$. Of course, in order for this to have a meaning, it is necessary that $(\Phi,X\varphi)\in\M_e$. Under this condition, we automatically have that $(X^\ddagger\Phi,\varphi)\in\M_e$. Notice further that, if $\ddagger$ is defined on $\V'$, it is also defined on $\V$. It is clear that, if $\V=\V'=\ltwo$ and $X\in B(\ltwo)$, the set of bounded operators on $\ltwo$, then $X^\dagger=X^\ddagger$ since, in particular, $(\Phi,X\varphi)\in\M_e$ $\forall\Phi,\varphi\in\ltwo$.

It is not difficult now to check for instance that, if $\V_0$ is $e$-separating, and if $X$ and $\Phi$ are such that $(\Phi,X^\ddagger\varphi)\in\M_e$ $\forall \varphi\in\V_0$, then $(X^\ddagger)^\ddagger\Phi=X\Phi$. Under similar assumptions we also deduce that
$$
(X+Y)^\ddagger\Phi=X^\ddagger\Phi+Y^\ddagger\Phi, \qquad (XY)^\ddagger\Phi=Y^\ddagger X^\ddagger \Phi.
$$
This means that $\ddagger$ has close properties to those required to any adjoint map. Incidentally we observe that the fact that $\ddagger$ is not always defined suggests to consider this map as a sort of adjoint for unbounded operators.

It is now useful to discuss a few examples of the $\ddagger$ adjoint.

Let $c=\frac{1}{\sqrt{2}}\left(x+\frac{d}{dx}\right)$ be the well known annihilation operator acting, in particular, on $\Sc(\mathbb{R})$. Let us take $\F_e$ as the ONB $\{e_n(x)\}$ of eigenstates of the harmonic oscillator, see (\ref{41}) below. We know that, \cite{mess},
\be
ce_0=0, \quad ce_n=\sqrt{n}e_{n-1}, \,n\geq1, \qquad \mbox{and}\qquad c^\dagger e_n=\sqrt{n+1}\,e_{n+1}, \,n\geq0.
\label{ladder}\en
 Now we have
$$
\pin{c^\ddagger \Phi}{\varphi}_e=\pin{ \Phi}{c\varphi}_e=\sum_{n=0}^\infty\pin{\Phi}{e_n}\pin{e_n}{c\varphi}=$$
$$=\sum_{n=0}^\infty\pin{\Phi}{e_n}\pin{c^\dagger e_n}{\varphi}=\sum_{n=0}^\infty\sqrt{n+1}\,\pin{\Phi}{e_n}\pin{e_{n+1}}{\varphi},
$$ 
 for all $\varphi$ and $\Phi$ for which this series converges. It is now easy to check that $\pin{c^\dagger \Phi}{\varphi}_e$ produces exactly the same result. Hence we conclude that (with a simplifying notation) $c^\ddagger=c^\dagger$. 
 
 With very similar computations we can also check that $(c^\ddagger)^\ddagger=c$:
 $$
 \pin{(c^\ddagger)^\ddagger \Phi}{\varphi}_e=\pin{ \Phi}{\varphi}_e,
 $$
for all $\Phi$ and $\varphi$ for which these quantities make sense. Formally, this can be seen simply as a consequence of the previous equality: $(c^\ddagger)^\ddagger=(c^\dagger)^\dagger=c$.

It is now clear that
$$
\left(\frac{d}{dx}\right)^\ddagger=-\frac{d}{dx}, \qquad \mbox{and}\qquad x^\ddagger=x.
$$
Then, at least for these operators, the adjoint $\ddagger$ is not really different from the usual one, $\dagger$. Hence we could say, in particular, that $x$ and $p=-i\,\frac{d}{dx}$ are {\em $\ddagger$-Hermitian} operators.

\section{Examples}\label{sect4}

In this section we will discuss a series of examples showing that, in some relevant cases, the $e$-product can be explicitly computed even outside $\ltwo$ and that, when this is possible, and when it is also possible to compute the integrals of the product of these two functions, the results coincide. 

In what follows we will always use, as our ONB $\F_e$, the eigenstates of the quantum harmonic oscillator already mentioned several times before in this paper:
\be
e_n(x)=\frac{1}{\sqrt{2^n\,n!\,\sqrt{\pi}}}\,H_n(x)\,e^{-x^2/2},
\label{41}\en
where $n=0,1,2,3,\ldots$, and $H_n(x)$ is the $n$-th Hermite polynomial. This is not the only choice, and it is not necessarily the most convenient in all the applications. In some other cases it could be possible that a different choice turns out to be more suitable for the particular situation we are interested in. However, these $e_n(x)$'s make it possible to compute all the relevant integrals in the examples considered below, and this is quite useful for us. We will comment on other possible choices of the ONB later on.

\subsection{Example nr. 1}

We first show how to compute the $e$-product of $F(x)=e^{\gamma x}$ and $G(x)=\delta(x)$, $\gamma\in\mathbb{R}$. It is clear that neither $F(x)$ nor $G(x)$ are square-integrable. In particular, $G(x)$ is not even a function, while $F(x)$  exponentially diverges for $x\rightarrow\pm\infty$ depending on the sign of $\gamma$. Nevertheless, since $F(x)$ is continuous, we can compute the following integral
$$
\int_{\mathbb{R}} F(x)\,G(x)\,dx=\int_{\mathbb{R}} e^{\gamma x}\,\delta(x)\,dx=1.
$$
With a little abuse of language, this integral could be called {\em the scalar product} between $F(x)$ and $G(x)$, even if only in an extended sense.

Let now use (\ref{34}) in our case. We have
\be
\pin{F}{G}_e=\sum_{n=0}^\infty{F}[e_n]G[e_n],
\label{42}\en
noticing that $F(x)=\overline{F(x)}$. It is clear that both ${F}[e_n]=\pin{F}{e_n}$ and $G[e_n]=\pin{e_n}{\delta}$ are well defined for all $n\geq0$. What is not so clear is whether the series converges or not. In fact, this is what we will show now, starting with the following obvious result:
\be
\pin{e_n}{\delta}=e_n(0)=\frac{1}{\sqrt{2^n\,n!\,\sqrt{\pi}}}\,H_n(0),
\label{43}\en
for all $n\geq0$. As for $F[e_n]$ we have
$$
{F}[e_n]=\pin{F}{e_n}=\frac{1}{\sqrt{2^n\,n!\,\sqrt{\pi}}}\int_{\mathbb{R}}e^{\gamma x}H_n(x)e^{-x^2/2}\,dx=\frac{\sqrt{2\pi}\,i^n}{\sqrt{2^n\,n!\,\sqrt{\pi}}}\,H(-i\gamma)e^{\gamma^2/2},
$$
using the following integral, \cite{grad}:
\be
\int_{\mathbb{R}}e^{ixy}H_n(x)e^{-x^2/2}\,dx=\sqrt{2\pi}\,i^n\,e^{-y^2/2}H_n(y).
\label{44}\en
Hence we have
$$
\pin{F}{G}_e=\sqrt{2}\,e^{\gamma^2/2}\sum_{n=0}^\infty\frac{1}{n!}\left(\frac{i}{2}\right)^nH_n(-i\gamma)H_n(0)=1,
$$
where we have used the Mehler's formula for the Hermite polynomials, \cite{mehler}:
\be
\sum_{n=0}^\infty\frac{z^n}{n!}H_n(x)H_n(y)=\left(1-4z^2\right)^{-1/2}\exp\left\{y^2-\frac{(y-2zx)^2}{1-4z^2}\right\},
\label{45}\en
for all $z\neq\frac{1}{2}$.

The conclusion is that the $e$-product exists in this case, and that it returns the expected result.

\subsection{Example nr. 2}

For our second example we take $F(x)=\cos(x)$ and $G(x)=\delta(x)$. Again, neither $F(x)$ nor $G(x)$ are square-integrable but the integral between the two can be computed since $F(x)$ is continuous, and we get
$$
\int_{\mathbb{R}} F(x)\,G(x)\,dx=\int_{\mathbb{R}} \cos(x)\,\delta(x)\,dx=1.
$$
As in the previous example we have $\pin{F}{G}_e=\sum_{n=0}^\infty{F}[e_n]G[e_n]$, and $G[e_n]=e_n(0)$. Slightly longer is the computation of $F[e_n]$, but is not complicated at all. We first observe that
$$
{F}[e_n]=\pin{F}{e_n}=\frac{1}{2}\left(c_n+\overline{c_n}\right), \qquad c_n=\int_{\mathbb{R}}e^{ix}e_n(x)dx.
$$
We see that $c_n$ is proportional to the Fourier transform of $e_n(x)$ computed in $p=1$. More explicitly, we can use the formula in (\ref{44}).  We deduce that
$$
F[e_{2l+1}]=0, \qquad F[e_{2l}]=\sqrt{\frac{\sqrt{\pi}}{2^{2l-1}(2l)!e}}(-1)^lH_{2l}(1),
$$
$\forall l\geq0$. Hence we find that
$$
\pin{F}{G}_e=\sqrt{\frac{2}{e}}\sum_{l=0}^\infty\frac{(-i)^{2l}}{2^{2l}(2l)!}H_{2l}(1)H_{2l}(0),
$$
which can be computed using  formula (\ref{45}). In fact, we can add to the sum above all the {\em odd terms},
$$
\sum_{l=0}^\infty\frac{(-i)^{2l+1}}{2^{2l+1}(2l+1)!}H_{2l+1}(1)H_{2l+1}(0),
$$
since this sum is zero due to the fact that $H_{2l+1}(0)=0$ for all $l\geq0$. In this way we get
$$
\sum_{l=0}^\infty\frac{(-i)^{2l}}{2^{2l}(2l)!}H_{2l}(1)H_{2l}(0)=\sum_{l=0}^\infty\frac{(-i)^{l}}{2^{l}(l)!}H_{l}(1)H_{l}(0)=\sqrt{\frac{e}{2}},
$$
so that $\pin{F}{G}_e=1$, as expected.

We could easily adapt these computations to check that, putting $L(x)=\sin(x)$, $\pin{L}{G}_e=0$, once again in agreement with the fact that $\int_{\mathbb{R}} L(x)\,G(x)\,dx=L(0)=0$.

\subsection{Example nr. 3}

Let us now take $F(x)=G(x)=\delta(x)$. In this case, of course, $\int_{\mathbb{R}} F(x)\,G(x)\,dx$ does not exist. Still we could try to see if $\pin{\delta}{\delta}_e$ makes sense or not. Unfortunately, the result will be that in this case the series defining $\pin{\delta}{\delta}_e$ does not converge. The positive aspect of this result, however, is that we will easily extend our conclusion to (weak) derivatives of $\delta(x)$, showing that in some cases the $e$-product can be easily computed. 

Let us start with $\pin{\delta}{\delta}_e$.
In this case we have
$$
\pin{\delta}{\delta}_e=\sum_{n=0}^\infty(e_n(0))^2=\frac{1}{\sqrt{\pi}}\sum_{l=0}^\infty\frac{H_{2l}^2(0)}{2^{2l}(2l)!}=\frac{1}{\sqrt{\pi}}\sum_{l=0}^\infty\frac{(2l)!}{4^l(l!)^2},
$$
since $H_{2l+1}(0)=0$ for all $l\geq0$. We have used here the following formula for the Hermite polynomial: $H_{2l}(0)=(-1)^l\frac{(2l)!}{l!}$. To study the convergence of the series one could use the ratio test, but this does not allow us to conclude: if we put  $a_l=\frac{(2l)!}{4^l(l!)^2}$ we easily find that $\lim_{l,\infty}\frac{a_{l+1}}{a_l}=1$. However, we can also deduce that
$$
\lim_{l,\infty}\left[l\frac{a_l}{a_{l+1}}-1\right]=\frac{1}{2}.
$$
Since this limit is less than one, the Raabe test allows us to conclude that $\sum_{l=0}^\infty\frac{(2l)!}{4^l(l!)^2}$ diverges, so that $\pin{\delta}{\delta}_e$ does not exist. Of course, this does not imply, in general, that $\pin{\delta}{\delta}_c$ exists finite for some different choice of ONB $\F_c$. We will comment on that later.
Here, what we can check very easily is the following rather general result:
\be
\pin{\delta^{(k)}}{\delta^{(l)}}_e=\left\{
\begin{array}{ll}
0 \hspace{3cm} \mbox  { if } k+l \mbox{ is odd},\\
\nexists\hspace{3cm} \mbox  { if }  k+l \mbox{ is even}.
\end{array}
\right.
\label{46}\en
Here $\delta^{(n)}(x)$ is the $n$-th weak derivative of the delta distribution. Of course this result extends the one above, which is recovered for $k=l=0$. The proof of the first result is easy, and follows from the following equality: $$\pin{\delta^{(k)}}{\delta^{(l)}}_e=\sum_{n=0}^\infty\pin{\delta^{(k)}}{e_n}\pin{e_n}{\delta^{(l)}}=(-1)^{k+l}\sum_{n=0}^\infty\pin{\delta}{e_n^{(k)}}\pin{e_n^{(l)}}{\delta}=(-1)^{k+l}\sum_{n=0}^\infty e_n^{(k)}(0)e_n^{(l)}(0).$$
It is now clear that, if $k+l$ is odd, either $k$ or $l$ is odd. Suppose $k$ is odd. Of course, $l$ must be even. Now, due to the well known parity properties of $e_n(x)$, expressed by $e_n(-x)=(-1)^ne_n(x)$, we have that $e_n^{(k)}(0)=0$ for all even $n$. Hence $e_{2j}^{(k)}(0)e_{2j}^{(l)}(0)=0$ for all $j=0,1,2,\ldots$. Similarly, and for the same $j$, $e_{2j+1}^{(k)}(0)e_{2j+1}^{(l)}(0)=0$ since $e_{2j+1}^{(l)}(0)=0$. This implies that $\pin{\delta^{(k)}}{\delta^{(l)}}_e=0$ in this case. The same idea works if $k$ is even and $l$ is odd.

This argument does not work if $k+l$ is even, i.e. when $k$ and $l$ are both even, or both odd. In the first case we find that $\pin{\delta^{(k)}}{\delta^{(l)}}_e=\sum_{n=0}^\infty e_{2n}^{(k)}(0)e_{2n}^{(l)}(0)$, while in the second case we get $\pin{\delta^{(k)}}{\delta^{(l)}}_e=\sum_{n=0}^\infty e_{2n+1}^{(k)}(0)e_{2n+1}^{(l)}(0)$. 

To check that $\pin{\delta^{(1)}}{\delta^{(1)}}_e$ does not exist, we first notice that, since the bosonic lowering operator $c$ introduced in (\ref{ladder}) can be written as $c=\frac{1}{\sqrt{2}}\left(x+\frac{d}{dx}\right)$, then
$$
\frac{d}{dx}e_{2n+1}(x)=\sqrt{2}\left(c-\frac{x}{\sqrt{2}}\right)e_{2n+1}(x)=\sqrt{2}\left(\sqrt{2n+1}e_{2n}(x)-\frac{x}{\sqrt{2}}e_{2n+1}(x) \right),
$$
using (\ref{ladder}),
if $n\geq1$, so that
$$
e_{2n+1}^{(1)}(0)=\sqrt{2(2n+1)}e_{2n}(0).
$$
Then we have 
$$
\pin{\delta^{(1)}}{\delta^{(1)}}_e=\sum_{n=0}^\infty e_{2n+1}^{(1)}(0)e_{2n+1}^{(1)}(0)=2\sum_{n=0}^\infty(2n+1)(e_{2n}(0))^2\geq 2\sum_{n=0}^\infty(e_{2n}(0))^2=2\pin{\delta}{\delta}_e,
$$
which, as we have seen, does not exist. It is not hard to imagine that, using a similar trick, it is also possible to estimate other  $\pin{\delta^{(k)}}{\delta^{(k)}}_e$ with $k\geq2$ in terms of $\pin{\delta}{\delta}_e$. We expect that similar estimates also extend to $\pin{\delta^{(k)}}{\delta^{(l)}}_e$, where $k$ and $l$ are different but have the same parity\footnote{In this sense the second line in (\ref{46}), more than being proved in general, is indeed an expected result.}.

\subsection{Example nr. 4}

This is maybe the most interesting example here, at least for us, since it is exactly what motivated us in \cite{bag2020JPA} to extend the ordinary scalar product from $\ltwo$ to distributions, as shown in (\ref{21}). It might be useful to briefly sketch the physical motivations of this particular example.

We start introducing the operators  $\hat xf(x)=xf(x)$ and $(\hat Dg)(x)=g'(x)$, the derivative of $g(x)$, for all $f(x)\in D(\hat x)=\{h(x)\in\Lc^2(\mathbb{R}): xh(x)\in \Lc^2(\mathbb{R} \}$ and $g(x)\in D(\hat D)=\{h(x)\in\Lc^2(\mathbb{R}: h'(x)\in \Lc^2(\mathbb{R} \}$.  The adjoints of $\hat x$ and $\hat D$ in $\ltwo$ are  $\hat x^\dagger=\hat x$, $\hat D^\dagger=-\hat D$, and we have $[D,x]f(x)=f(x)$, for all $f(x)\in\Sc(\mathbb{R})$, which is contained in $D(\hat x)\cap D(\hat D)$. If we look for the vacua of $a=\hat D$ and $b=\hat x$, \cite{bag2020JPA}, we easily find that these are $\varphi_0(x)=1$ and $\psi_0(x)=\delta(x)$, with a suitable choice of the normalization. It is clear, therefore, that neither $\varphi_0(x)$ nor $\psi_0(x)$ belong to $\Sc(\mathbb{R})$, or even to $\ltwo$. However we can still put
\be
\varphi_n(x)=\frac{b^n}{\sqrt{n!}}\,\varphi_0(x)=\frac{x^n}{\sqrt{n!}}, \qquad \psi_n(x)=\frac{(a^\dagger)^n}{\sqrt{n!}}\,\psi_0(x)=\frac{(-1)^n}{\sqrt{n!}}\,\delta^{(n)}(x),
\label{48}\en
for all $n=0,1,2,3,\ldots$. Here $\delta^{(n)}(x)$ is the n-th weak derivative of the Dirac delta function and it turns out that $a$, $b$ and their adjoints work as ladder operators on $\F_\varphi=\{\varphi_n(x)\}$ and  $\F_\psi=\{\psi_n(x)\}$, and that calling $N=ba$ we have
\be
N\varphi_k(x)=k\varphi_{k}(x), \qquad \qquad N^\dagger\psi_k(x)=k\psi_{k}(x),
\label{49}\en
for all $k=0,1,2,3,\ldots$. This is a non trivial example of what we have called weak pseudo-bosons in \cite{bag2020JPA}.

In that paper we were able to prove that an extended scalar product between $\varphi_k(x)$ and $\psi_k(x)$ can be defined, as in (\ref{21}) and that
\be\pin{\varphi_n}{\psi_k}=\delta_{n,k},
\label{47}\en
$\forall n,k\geq0$. Here $\delta_{n,k}$ is the Kronecker delta. These and other results are given in \cite{bagspringer,bag2020JPA}. What we want to do here is to show that the same biorthonormality result can be deduced by replacing $\pin{.}{.}$ with $\pin{.}{.}_e$.

We have:
\be
\pin{\varphi_n}{\psi_m}_e=\frac{1}{\sqrt{n!\,m!}}\pin{x^n}{\delta^{(m)}}_e=\frac{1}{\sqrt{n!\,m!}}\sum_{k=0}^\infty\pin{x^n}{e_k}\pin{e_k}{\delta^{(m)}}.
\label{410}\en
Now, using the integral in (\ref{44}), we have
$$
\int_{\mathbb{R}}e^{-t^2/2}H_k(t)\delta^{(m)}(t)\,dt=(-1)^m\int_{\mathbb{R}}\frac{d^m}{dt^m}\left(e^{-t^2/2}H_k(t)\right)\delta(t)\,dt=
$$
$$
=(-1)^m\,\frac{(-i)^k}{\sqrt{2\pi}}\int_{\mathbb{R}}\frac{d^m}{dt^m}\left(\int_{\mathbb{R}}e^{ixt}e^{-x^2/2}H_k(x)\right)\delta(t)\,dt=\frac{(-i)^{m+k}}{\sqrt{2\pi}}\int_{\mathbb{R}}x^me^{-x^2/2}H_k(x)\,dx.
$$
Hence we find
\be
\pin{\varphi_n}{\psi_m}_e=\frac{(-i)^m}{\pi\,\sqrt{2\,n!\,m!}}\sum_{k=0}^\infty \frac{(-i)^k}{2^k\,k!}\,I_k(n)\,I_k(m),
\label{411}\en
where
\be
I_k(p)=\int_{\mathbb{R}}x^pe^{-x^2/2}H_k(x)\,dx,
\label{412}\en
for $p,k\geq0$. To compute now (\ref{411}) it is useful to notice first that, due to the parity of the Hermite polynomials,
\be
I_{2n}(2l+1)=I_{2n+1}(2l)=0
\label{413}\en
for all possible values of $n$ and $l$ in $\mathbb{N}_0=\mathbb{N}\cup\{0\}$. Hence we have, clearly,
\be
\pin{\varphi_{2n}}{\psi_{2m+1}}_e=\pin{\varphi_{2n+1}}{\psi_{2m}}_e=0,
\label{414}\en
$\forall n,m\in\mathbb{N}_0$. Slightly more complicated is the computation of $\pin{\varphi_{2n}}{\psi_{2m}}_e$ and $\pin{\varphi_{2n+1}}{\psi_{2m+1}}_e$. In these cases we need the following formulas, see \cite{grad}:
$$
\int_0^\infty e^{-2\alpha x^2}x^\nu H_{2r}(x)\,dx= \frac{(-1)^r}{\sqrt{\pi}\,\alpha^{(\nu+1)/2}} 2^{2r-\frac{3}{2}-\frac{1}{2}\nu}\,\Gamma\left(\frac{\nu+1}{2}\right)\Gamma\left(r+\frac{1}{2}\right)F\left(-r,\frac{\nu+1}{2};\frac{1}{2};\frac{1}{2\alpha}\right)
$$
where $\Re(\alpha)>0$ and $\Re(\nu)>-1$, and
$$
\int_0^\infty e^{-2\alpha x^2}x^\nu H_{2r+1}(x)\,dx= \frac{(-1)^r}{\sqrt{\pi}\,\alpha^{\nu/2+1}} 2^{2r-\frac{1}{2}\nu}\Gamma\left(\frac{\nu}{2}+1\right)\Gamma\left(r+\frac{3}{2}\right)F\left(-r,\frac{\nu}{2}+1;\frac{3}{2};\frac{1}{2\alpha}\right)
$$
where, again, $\Re(\alpha)>0$ with $\Re(\nu)>-2$. Here $\Gamma$ is the Gamma function, while $F$ is the Gauss hypergeometric function. After some computations we find
\be
\pin{\varphi_{2n}}{\psi_{2m}}_e=\frac{(-1)^m}{\pi^2\sqrt{\,n!\,m!}}\,2^{n+m+1/2}\,\Gamma\left(n+\frac{1}{2}\right)\Gamma\left(m+\frac{1}{2}\right)\sum_{j=0}^\infty a_j(n,m),
\label{415}\en
and 
\be
\pin{\varphi_{2n+1}}{\psi_{2m+1}}_e=\frac{(-1)^{m+1}}{\pi^2\sqrt{\,n!\,m!}}\,2^{n+m+11/2}\,\Gamma\left(n+\frac{3}{2}\right)\Gamma\left(m+\frac{3}{2}\right)\sum_{j=0}^\infty b_j(n,m),
\label{416}\en
where we have introduced
\be
a_j(n,m)=\frac{(-4)^j}{(2j)!}\left(\Gamma\left(j+\frac{1}{2}\right)\right)^2F\left(-j,n+\frac{1}{2};\frac{1}{2};2\right)F\left(-j,m+\frac{1}{2};\frac{1}{2};2\right),
\label{417}\en
and
\be
b_j(n,m)=\frac{(-4)^j}{(2j+1)!}\left(\Gamma\left(j+\frac{3}{2}\right)\right)^2F\left(-j,n+\frac{3}{2};\frac{3}{2};2\right)F\left(-j,m+\frac{3}{2};\frac{3}{2};2\right).
\label{418}\en

\subsubsection{An introductory case: $n=m=0$}

We start computing $\pin{\varphi_{0}}{\psi_{0}}_e$ using (\ref{415}) and (\ref{417}). We have
$$
\pin{\varphi_{0}}{\psi_{0}}_e=\frac{1}{\pi^2}\,2^{1/2}\,\Gamma\left(\frac{1}{2}\right)\Gamma\left(\frac{1}{2}\right)\sum_{j=0}^\infty a_j(0,0),
$$
with
$$
a_j(0,0)=\frac{(-4)^j}{(2j)!}\left(\Gamma\left(j+\frac{1}{2}\right)\right)^2F\left(-j,\frac{1}{2};\frac{1}{2};2\right)F\left(-j,\frac{1}{2};\frac{1}{2};2\right)=\frac{(-4)^j}{(2j)!}\left(\Gamma\left(j+\frac{1}{2}\right)\right)^2
$$
 since, \cite{grad}, $F\left(-j,\frac{1}{2};\frac{1}{2};2\right)=(1-2)^j=(-1)^j$, for all $j\geq0$. The sum of the series can be computed, and we get (using, e.g., Mathematica)
 $$
 \sum_{j=0}^\infty a_j(0,0)=\frac{\pi}{\sqrt{2}}.
 $$
Putting all together we easily conclude now that $\pin{\varphi_{0}}{\psi_{0}}_e=1$.

Now, to compute  $\pin{\varphi_{1}}{\psi_{1}}_e$, we use (\ref{416}) and (\ref{418}), and we get
$$
\pin{\varphi_{1}}{\psi_{1}}_e=\frac{2^{11/2}}{\pi^2}\,\Gamma\left(\frac{3}{2}\right)\Gamma\left(\frac{3}{2}\right)\sum_{j=0}^\infty b_j(0,0),
$$
where
$$
b_j(0,0)=\frac{(-4)^j}{(2j+1)!}\left(\Gamma\left(j+\frac{3}{2}\right)\right)^2(F\left(-j,\frac{3}{2};\frac{3}{2};2\right))^2=\frac{(-4)^j}{(2j+1)!}\left(\Gamma\left(j+\frac{3}{2}\right)\right)^2,
$$
since we have  $F\left(-j,\frac{3}{2};\frac{3}{2};2\right)=(1-2)^j=(-1)^j$. If we now compute $ \sum_{j=0}^\infty b_j(0,0)$ (using again Mathematica), the result is that the series does not converge. However, we can show that the series is actually Abel-convergent. Indeed, let us consider the following power series:
$$
S(z)=\sum_{j=0}^\infty \frac{(\Gamma(j+3/2))^2}{(2j+1)!}\,z^j.
$$
It is clear that we can recover $ \sum_{j=0}^\infty b_j(0,0)=\lim_{z,-4}S(z)$, if this limit exists. Indeed it is easy to check that the radius of convergence of $S(z)$ is $\rho_S=4$, so that $z=-4$ is on the boundary. Now, while $\lim_{z,4}S(z)$ does not exist, we find that $\lim_{z,-4}S(z)=\frac{\pi}{8\,\sqrt{2}}$. We finally conclude that
$$
\pin{\varphi_{1}}{\psi_{1}}_e=\frac{2^{11/2}}{\pi^2}\,\left(\frac{\sqrt{\pi}}{2}\right)^2\frac{\pi}{8\,\sqrt{2}}=1,
$$
as expected.

In other terms we have
$$
\pin{1}{\delta}_e=\pin{x}{-\delta'}_e=1.
$$

\subsubsection{Larger $n$ and $m$}

For larger values of $n$ and $m$ it is convenient to consider the same Abel summation we have used for $\pin{\varphi_{1}}{\psi_{1}}_e$. The main problem in (\ref{415}) and (\ref{416}) is to compute $\sum_{j=0}^\infty a_j(n,m)$ and $\sum_{j=0}^\infty b_j(n,m)$. We use the same idea as for $n=m=0$, introducing
$$
\sigma(z;2n,2m)=\sum_{j=0}^\infty \alpha_j(n,m) z^j, \qquad \sigma(z;2n+1,2m+1)=\sum_{j=0}^\infty \beta_j(n,m) z^j, 
$$
where
$$
\alpha_j(n,m)=\frac{1}{(2j)!}\left(\Gamma\left(j+\frac{1}{2}\right)\right)^2F\left(-j,n+\frac{1}{2};\frac{1}{2};2\right)F\left(-j,m+\frac{1}{2};\frac{1}{2};2\right)
$$
and
$$
\beta_j(n,m)=\frac{1}{(2j+1)!}\left(\Gamma\left(j+\frac{3}{2}\right)\right)^2F\left(-j,n+\frac{3}{2};\frac{3}{2};2\right)F\left(-j,m+\frac{3}{2};\frac{3}{2};2\right).
$$
Once again, a direct computation of $\sigma(-4;2n,2m)$ and $\sigma(-4;2n+1,2m+1)$ with Mathematica produces a non-convergence warning message. However, if we compute $\lim_{z,-4}\sigma(z;2n,2m)$ and $\lim_{z,-4}\sigma(z;2n+1,2m+1)$ we conclude that, as expected,
\be
\pin{\varphi_n}{\psi_m}_e=\delta_{n,m}.
\label{419}\en
Indeed we have no rigorous prove of this result for all $n,m\in\mathbb{N}_0$. However, we have performed the check of this result using Mathematica in hundreds of cases, fixing different values of $n$
 and $m$, and all the outputs of these checks are compatible with (\ref{419}). Of course this is only an indication that (\ref{419}) holds true for all $n,m\in\mathbb{N}_0$, but (in our opinion) it is really a {\bf strong} indication! We hope to be able to provide an analytical proof of (\ref{419}) in a close future. 
 
 \vspace{2mm}
 
 {\bf Remark:--} in these computations the role of the Abel convergence is essential. This is interesting, and it was already noticed in \cite{bag2023}. We wonder if there is some major relation between the $e$-product of certain distributions and the necessity of using this kind of convergence. This is work in progress.

\subsection{Two variations on the same theme}

Let us first try to compute $\pin{\varphi_n}{\varphi_m}_e$, where $\varphi_n(x)=\frac{x^n}{\sqrt{n!}}$ is the function introduced in (\ref{48}). Using the definition in (\ref{415}) we find that
\be
\pin{\varphi_n}{\varphi_m}_e=\frac{1}{\sqrt{\pi\,n!\,m!}}\sum_{k=0}^\infty \frac{1}{2^k\,k!}\,I_k(n)\,I_k(m).
\label{420}\en
As in (\ref{414}) we can easily see that
\be
\pin{\varphi_{2n}}{\varphi_{2m+1}}_e=0,
\label{421}\en
$\forall n,m\in\mathbb{N}_0$. With similar computations as before we find that, see (\ref{415}) and (\ref{416}),

\be
\pin{\varphi_{2n}}{\varphi_{2m}}_e=\frac{2^{n+m+1}}{\sqrt{\pi^3\,(2n)!\,(2m)!}}\,\Gamma\left(n+\frac{1}{2}\right)\Gamma\left(m+\frac{1}{2}\right)\sum_{j=0}^\infty c_j(n,m),
\label{422}\en
and 
\be
\pin{\varphi_{2n+1}}{\varphi_{2m+1}}_e=\frac{2^{n+m+6}}{\sqrt{\pi^3\,(2n+1)!\,(2m+1)!}}\,\Gamma\left(n+\frac{3}{2}\right)\Gamma\left(m+\frac{3}{2}\right)\sum_{j=0}^\infty d_j(n,m),
\label{423}\en
where
\be
c_j(n,m)=\frac{4^j}{(2j)!}\left(\Gamma\left(j+\frac{1}{2}\right)\right)^2F\left(-j,n+\frac{1}{2};\frac{1}{2};2\right)F\left(-j,m+\frac{1}{2};\frac{1}{2};2\right),
\label{424}\en
and
\be
d_j(n,m)=\frac{4^j}{(2j+1)!}\left(\Gamma\left(j+\frac{3}{2}\right)\right)^2F\left(-j,n+\frac{3}{2};\frac{3}{2};2\right)F\left(-j,m+\frac{3}{2};\frac{3}{2};2\right).
\label{425}\en
We observe that the only difference between $(c_j(n,m),d_j(n,m))$ and $(a_j(n,m),b_j(n,m))$ is that $-4$ is replaced by 4. Then we could construct, as we did before for $\sigma(z;2n,2m)$ and $\sigma(z;2n+1,2m+1)$, two power series and look for the radius of convergence. We can check that this radius is 4 for both the series. The point is that, while the limits for $z\rightarrow -4$ of $\sigma(z;2n,2m)$ and $\sigma(z;2n+1,2m+1)$ exist finite, the limits for $z\rightarrow 4$ of $\sigma(z;2n,2m)$ and $\sigma(z;2n+1,2m+1)$ do not. This means that $\pin{\varphi_{k}}{\varphi_{l}}_e$ does not exist for our choice of $\F_e$ if $k$ and $l$ have the same parity, while it is zero if they have opposite parity. Of course, this result leaves open the possibility of having a finite result replacying $\F_e$ with a different, and possibly more convenient, ONB. This possibility is work in progress.

\vspace{3mm}

If  instead of considering $\pin{\varphi_n}{\varphi_m}_e$ we focus on $\pin{\psi_n}{\psi_m}_e$, see (\ref{48}), we find that, with more or less the same computations,
\be
\pin{\psi_n}{\psi_m}_e=\frac{(-i)^{n+m}}{2\pi\,\sqrt{\pi\,n!\,m!}}\sum_{k=0}^\infty \frac{(-1)^k}{2^k\,k!}\,I_k(n)\,I_k(m),
\label{426}\en
and the conclusion is similar to the previous one: $\pin{\psi_n}{\psi_m}_e=0$ if $n$ and $m$ have opposite parity, while it does not exist if $n$ and $m$ have the same parity. In fact, we can check that
$$
\pin{\psi_{2n}}{\psi_{2m}}_e=2\pi(-1)^{n+m}\pin{\varphi_{2n}}{\varphi_{2m}}_e,
$$
and 
$$
\pin{\psi_{2n+1}}{\psi_{2m+1}}_e=2\pi(-1)^{n+m}\pin{\varphi_{2n+1}}{\varphi_{2m+1}}_e.
$$
It might be worth noticing that the fact that a scalar product involving $\varphi_n(x)$ and $\psi_m(x)$ can only be safely introduced between some $\varphi_n(x)$ and some $\psi_m(x)$, but not between some $\varphi_n(x)$ and some other $\varphi_m(x)$ (or between some $\psi_n(x)$ and some other $\psi_m(x)$), is not a big surprise. In fact, because of their nature of generalized eigenstates of two number-like operators, see (\ref{49}), connected by the adjoint operation, it is somehow expected that they satisfy a sort of biorthonormality condition, even if in an extended sense, as in our case. On the other hand, there is no reason for the two families $\F_\varphi=\{\varphi_n(x)\}$ and $\F_\psi=\{\psi_n(x)\}$ to be orthonormal families, and in fact this is not true, even with respect to the $e$-product.

\section{Conclusions}\label{sect5}

In this paper we continue our analysis on the $\pin{.}{.}_e$ product we have recently proposed in connection with a gain and loss system. We discuss some of its properties, and we also briefly consider the possibility of considering an adjoint operation connected to $\pin{.}{.}_e$. Then we show how this product can be applied to functions which are not in $\ltwo$, and to distributions. In particular we have seen that a biorthonormality result can be found for the eigenstates of number-like operators constructed, following the general idea of weak pseudo-bosons, starting from two different vacua and acting on them with powers of suitable raising operators, which are essentially a multiplication and a derivation operator.

There are several ways to extend and enrich the analysis discussed here. First of all, it would be very interesting to check how much changing the ONB affects the result of our computations. Also, we have seen in the paper that several  $\pin{f}{g}_e$ do not exist (meaning that the series which defines the scalar product do not converge) for the choice of $\F_e$ we have considered all along this paper. But it remains open the possibility of choosing a different ONB, $\F_c$, such that $\pin{f}{g}_c$ does exist. This is part of our projects for this line of research, together with a deeper analysis of the properties of the adjoint $\ddagger$. The role of the Abel convergence is also a fascinating aspect of the $e$-product, and should be understood. Last but not least, we hope to propose more physical applications of the $e$-product, for some relevant concrete system.

\section*{Acknowledgements}

The author acknowledges partial financial support from Palermo University and from G.N.F.M. of the INdAM.  This work has also been partially supported by the PRIN grant {\em Transport phenomena in low dimensional
	structures: models, simulations and theoretical aspects}- project code 2022TMW2PY - CUP B53D23009500006, and partially by  the project ICON-Q, Partenariato Esteso NQSTI - PE00000023, Spoke 2.

\section*{Data accessibility statement}

This work does not have any experimental data.

%\section*{Competing interests statement}
%
%We have no competing interests.
%
%\section*{Authors' contributions}
%
%FB proposed the model and its solution. AK and EH worked on the
%interpretation of the model in the context of DM.
%
%\section*{Acknowledgements}
%
%One of the authors (FB) acknowledges partial support from the University of
%Palermo and from G.N.F.M. and the discussion on the paper was started during
%his visit to Linnaeus University (supported by the  grant ``Mathematical
%Modeling of Complex Hierarchic Systems''). FB also wishes to thank Prof. A.
%Busacca for his strong support for this project.

\section*{Funding statement}

The author acknowledges partial financial support from Palermo University and from G.N.F.M. of the INdAM. This work has also been partially supported by the PRIN grant {\em Transport phenomena in low dimensional
	structures: models, simulations and theoretical aspects}.

\end{document}